\documentclass[11pt, a4paper]{article}

\usepackage{amsmath,amsthm,amssymb,cancel,graphicx,url}
\usepackage[margin=2cm]{geometry}

\theoremstyle{plain}

\newtheorem{corollary}{Corollary}
\newtheorem{lemma}{Lemma}
\newtheorem{proposition}{Proposition}
\newtheorem{theorem}{Theorem}
\newtheorem*{claim*}{Claim}

\theoremstyle{definition}
\newtheorem{definition}{Definition}

\newtheorem{example}{Example}

\newtheorem{question}{Question}

\newcommand{\N}{\mathbb{N}}

\newcommand{\cst}{\mathrm{cst}}
\newcommand{\BF}[1]{{\bf\boldmath{#1}\unboldmath}}
\newcommand{\Fix}{\mathrm{Fix}}
\newcommand{\G}{\mathrm{G}}

\newcommand{\dH}{d_\mathrm{H}}
\newcommand{\dm}{d_\mathrm{m}}
\newcommand{\dM}{d_\mathrm{M}}
\newcommand{\AH}{A_\mathrm{H}}
\newcommand{\Am}{A_\mathrm{m}}
\newcommand{\AM}{A_\mathrm{M}}

\title{Fixed points of Boolean networks, guessing graphs, and coding theory}

\author{Maximilien Gadouleau\footnote{School of Engineering and Computing Sciences, Durham University, UK. \texttt{m.r.gadouleau@durham.ac.uk}}
\footnote{Corresponding author.} 
\and Adrien Richard\footnote{Laboratoire I3S, CNRS \& Universit\'e de Nice-Sophia Antipolis, France. \texttt{richard@unice.fr}} 
\and S\o ren Riis\footnote{School of Electronic Engineering and Computer Science, Queen Mary, University of London, UK. \texttt{s.riis@qmul.ac.uk}}}

\date{\today\footnote{This work is partially supported by CNRS and Royal Society through the International Exchanges Scheme grant {\em Boolean networks, network coding and memoryless computation.}}}

\begin{document}

\maketitle

\begin{abstract}
In this paper, we are interested in the number of fixed points of functions $f:A^n\to A^n$ over a finite alphabet $A$ defined on a given signed digraph $D$. We first use techniques from network coding to derive some lower bounds on the number of fixed points that only depends on $D$. We then discover relationships between the number of fixed points of $f$ and problems in coding theory, especially the design of codes for the asymmetric channel.  Using these relationships, we derive upper and lower bounds on the number of fixed points, which significantly improve those given in the literature. We also unveil some interesting behaviour of the number of fixed points of functions with a given signed digraph when the alphabet varies. We finally prove that signed digraphs with more (disjoint) positive cycles actually do not necessarily have functions with more fixed points.
\end{abstract}

\section{Introduction} \label{sec:intro}

Boolean networks have been used to represent a network of interacting agents as follows. A network of $n$ automata has a state $x= (x_0,\dots, x_{n-1}) \in \{0,1\}^n$, represented by a Boolean variable $x_i$ on each automaton $i$, which evolves according to a deterministic function $f = (f_0,\dots,f_{n-1}) : \{0,1\}^n \to \{0,1\}^n$, where $f_i : \{0,1\}^n \to \{0,1\}$ represents the update of the local state $x_i$. Boolean networks have been used to model gene networks \cite{Kau69, Tho73, TK01a, Jong02}, neural networks \cite{MP43, Hop82, Gol85}, social interactions \cite{PS83, GT83} and more (see \cite{TD90, Goles90}). Their natural generalisation where each variable $x_i$ can take more than two values in some finite alphabet $A$ has been investigated since this can be a more accurate representation of the phenomenon we are modelling \cite{TK01a,KS08}.

The structure of a network $f: A^n \to A^n$ can be represented via its interaction graph $G(f)$, which indicates which update functions depend on which variables. More formally, $G(f)$ has $\{0,\dots,n-1\}$ as vertex set and there is an arc from $j$ to $i$ if $f_i(x)$ depends essentially on $x_j$. The arcs of the interaction graph can also be signed, where the sign of $(j,i)$ indicates whether the local function $f_i(x)$ is an increasing (positive sign), decreasing (negative sign), or non-monotone (zero sign) function of $x_j$. This is commonly the case when studying gene networks, where a gene can typically either activate (positive sign) or inhibit (negative sign)  another gene. In this biological context, the interaction graph is known--or at least well approximated--, while the actual update functions are not. The main problem of research on (non-necessarily Boolean) networks is then to predict their dynamics according to their interaction graphs. 

Among the many dynamical properties that can be studied, fixed points are crucial because they represent stable states; for instance, in the context of gene networks, they correspond to stable patterns of gene expression at the basis of particular biological processes. As such, they are arguably the property which has been the most thoroughly studied. The study of the number of fixed points and its maximisation in particular is the subject of a stream of work, e.g. in \cite{R86,Aracena2004,RRT08,Ara08,Richard09,ARS14}. A lot of literature is devoted to determining when a Boolean network admits multiple fixed points (see \cite{Ric13} for a survey). 

For the maximum number of fixed points with a  given signed interaction graph, however, a wide gap remains between the best lower bounds and upper bounds known so far. The first result in this area, proposed by Thomas \cite{Tho80}, is that networks whose interaction graphs do not have a positive cycle (i.e. a cycle with an even number of negative arcs) have at most one fixed point. This was then generalised into an upper bound on the number of fixed points of Boolean networks: a network has at most $2^{k^+}$ fixed points, where $k^+$ is the minimum size of a positive feedback vertex set of its interaction graph \cite{ADG04,Ara08}. This upper bound was then refined via the use of local graphs \cite{SD05, RRT08,Richard09,Ric13}. On the other hand, a positive cycle admits a Boolean network with $2$ fixed points. A network with a large number of fixed points for a general signed digraph is then obtained by packing positive cycles; if $D$ has $c^+$ disjoint positive cycles, then the network has $2^{c^+}$ fixed points. This result is folklore and is the best known by the authors. All these results tend to suggest that positive cycles in the interaction graph produce a high number of fixed points.

A completely different approach has been developed for unsigned digraphs in the context of network coding \cite{Rii06, Rii07, Rii07a}. Network coding is a technique to transmit information through networks, which can significantly improve upon routing in theory \cite{ACLY00, YLCZ06}. At each intermediate node $v$, the received messages $x_{u_1}, \ldots, x_{u_k}$ are combined, and the combined message $f_v(x_{u_1},\ldots,x_{u_k})$ is then forwarded towards its destination. The main problem is to determine which functions can transmit the most information. In particular, the network coding solvability problem tries to determine whether a certain network situation, with a given set of sources, destinations, and messages, is {\em solvable}, i.e. whether all messages can be transmitted to their destinations. As shown in \cite{Rii07, Rii07a}, the solvability problem can be recast in terms of fixed points of (non-necessarily Boolean) networks. The so-called guessing number \cite{Rii07} of a digraph $D$ is the logarithm of the maximum number of fixed points over all networks $f$ whose interaction graph is a subgraph of $D$: $G(f) \subseteq D$. Then, a network coding instance is solvable if some digraph $D$ related to the instance admits a network with guessing number equal to the size of the minimum feedback vertex set. The guessing number approach is further developed in \cite{GR11}, where the so-called guessing graph is introduced. This technique then completely omits the local update functions and instead turns the problem of maximising fixed points into a purely coding theoretic problem. Based on this approach, numerous upper and lower bounds on the guessing number of unsigned digraphs have been derived (see \cite{GR11}).

In  this paper, we generalise and adapt the techniques developed from network coding and coding theory to tackle the problem of maximising fixed points for signed digraphs. Firstly, we generalise the concept of guessing graph to all signed digraphs in Theorem \ref{th:g=alpha}. This approach is valid for all alphabet sizes, and immediately yields two new lower bounds on the guessing number of signed digraphs (in Theorems \ref{th:bound_g1} and \ref{th:bound_g2}, respectively). Secondly, we discover several relationships between fixed points of networks defined on signed digraphs and codes correcting asymmetric errors \cite{Klo95}. Based on these relationships, we then obtain much stronger (upper and lower) bounds on the number of fixed points via bounds on codes for asymmetric errors in Theorems \ref{th:g<A} and \ref{th:g>A}. These relationships, and the fact that they are so fruitful, are slightly counter-intuitive because the vast majority of error-correcting codes are linear, and hence involve linear functions over finite fields, which are not monotonous and thus cannot be used for signed digraphs. More strikingly, we obtain our tightest bounds for digraphs in which all the arcs are signed positively, and hence where all the local update functions are monotonically increasing (see Theorem \ref{th:g<A} for instance). This illustrates the versatility of the whole guessing number approach.

Our results also illustrate some different behaviour which occurs only for signed digraphs. Indeed, we discover a digraph for which the guessing number over any finite alphabet of size $s \ge 3$ is higher than the limit of the guessing number when $s$ tends to infinity. This is unlike the unsigned case where the limit is always the supremum over all finite alphabets. Finally, by comparing the numbers of fixed points for the negative clique and the positive clique, we then show that positive cycles are not necessarily useful in order to create many fixed points. Indeed, we exhibit two signed digraphs $D_1$ and $D_2$, where $D_1$ has more positive cycles and more disjoint positive cycles than $D_2$, and yet $D_2$ has a higher guessing number. This goes against the common view described above, and is actually akin to a result in \cite{ARS14} on the number of fixed points of conjunctive networks, which is maximised for a disjoint union of negative triangles (see Section \ref{sec:positive} for a more elaborate discussion).

% These are Boolean networks where every local update function $f_i(x)$ is a conjunction of literals: a positive or negative sign on the arc $(j,i)$ indicates whether the literal is $x_j$ or $\neg x_j$. It is shown that the maximum number of fixed points of a conjunctive network without loops in its interaction graph is obtained by using a disjoint union of triangles, where all arcs are signed negatively. Therefore, maximising the number of fixed points in the conjunctive case goes against maximising the number of (disjoint) positive cycles.

%The sets of fixed points of some special signed digraphs satisfy some combinatorial properties close to those for classical error-correcting codes for the Boolean alphabet, however it would be an interesting combinatorial problem to search for good codes for the positive clique on higher alphabets.

The rest of this paper is organised as follows. Section \ref{sec:guessing} first reviews some concepts on signed digraphs, and then introduces their guessing numbers and their guessing graphs. Based on the guessing graph, it then derives some lower bounds on the guessing number. Section \ref{sec:ECC_signed} determines different relationships between codes for the asymmetric channel and sets of fixed points and obtains more bounds on the guessing number. Section \ref{sec:positive} then investigates the guessing number of signed cliques. Finally, Section \ref{sec:comparison_bounds} compares the different bounds we have derived in the earlier sections.

\section{Guessing graph of a signed digraph} \label{sec:guessing}

\subsection{Signed digraphs and their guessing numbers} \label{sec:signed_digraphs}

A {\bf signed digraph} is a labelled digraph $D=(V,E,\lambda)$ where $\lambda:E\to\{-1,0,1\}$; typically $V = \{0,1,\ldots,n-1\}$. We shall equivalently represent a sign as $\alpha \in \{- , 0, +\}$.  We say that a signed digraph $D$ is {\bf unate} if $0\not\in\lambda(E)$; if all signs are equal, we shall make that explicit in our notation: we say that $D^0$ is {\bf unsigned} if $\lambda=\cst=0$, $D^+$ is {\bf positive} if $\lambda = \cst = 1$, and $D^-$ is {\bf negative} if $\lambda = \cst = -1$. We denote by $d$, $\delta$ and $\Delta$ the average, minimal and maximal in-degree of $D$, respectively. For all $i\in V$ and $\alpha \in \{-,0,+\}$, we denote by $d^\alpha_i$ the number of in-neighbours $j$ of $i$ such that $(j,i)$ has sign $\alpha$; $d^\alpha$, $\delta^\alpha$ and $\Delta^\alpha$ are then defined similarly. We set $\delta^\pm := \min_i \{d^+_i + d^-_i\}$.

The subgraphs of signed digraphs are defined as such. Let $D = (V,E,\lambda)$ and $D' = (V', E', \lambda')$ be two signed digraphs and let $|D| = (V,E)$ and $|D'| = (V',E')$ be their corresponding unsigned versions. We say that $D$  is a subgraph of $D'$, and denote it as $D \subseteq D'$, if $|D|$ is a subgraph of $|D'|$ and if the following holds: for every arc $e \in E$, $\lambda'(e) = -1$ implies $\lambda(e) = -1$ and $\lambda'(e) = 1$ implies $\lambda(e) = 1$. We remark that there is no constraint on the sign of $e$ in $D$ if it has zero sign in $D'$.

The {\bf sign of a cycle} of $D$ is the product of the sign of its arcs; a {\bf non-negative (resp. non-positive)} cycle is a cycle of sign $\ge 0$ (resp. $\le 0$). The {\bf non-negative girth} of $D$ is the minimum length of a non-negative cycle in $D$ and is denoted as $\gamma^+$. The subgraph of $D
$ induced by a set of vertices $U$ is denoted $D[U]$. A set of vertices $U \subseteq V$ such that $D[V \setminus U]$ does not contain any non-negative cycle is referred to as a {\bf non-negative feedback vertex set}. The minimum cardinality of a non-negative feedback vertex set is denoted as $k^+$. 

We denote by $N(i)$ the in-neighborhood of a vertex $i$ in $D$. For any vertex $i$ and $\alpha\in\{-,0,+\}$, we denote by $N^\alpha(i)$ the set of $j\in N(i)$ such that $\lambda(j,i)=\alpha$. If $h$ is a map from $V$ to the parts of $V$ then $h(U)=\cup_{i\in U} h(i)$ for all $U\subseteq V$. For example, $D$ is unate if $N^0(V)=\emptyset$ and unsigned if $N^+(V) = N^-(V) = \emptyset$. An arc of the form $(u,u)$ is a {\bf loop} on the vertex $u$.  

For every positive integer $a$, $[a]$ denotes the interval $\{0,1,\dots,a-1\}$. Let $s\ge 2$ and $n\ge 1$ be positive integers. For every $x\in [s]^n$ we write $x=(x_0,x_1,\dots,x_{n-1})$. The restriction of $x$ to a set of indices $I\subseteq [n]$ is denoted as $x_I$. 

Let $f:[s]^n\to [s]^n$. The set of fixed points of $f$ is denoted as $\Fix(f)$. We denote by $G(f)$ the \BF{signed interaction graph of $f$} as follows: the vertex set is $[n]$; for all $i,j\in [n]$, there exists an arc from $j$ to $i$ if $f_i$ depends essentially on $x_j$; and the sign $\lambda(j,i)$ of every arc $(j,i)$ is defined by 
\[
	\lambda(j,i)=
	\begin{cases}
	1&\text{if }f_i(x)\le f_i(x_0,\dots,x_j+1,\dots,x_{n-1})\text{ for all }x\in [s]^n\text{ with }x_j<s-1\\
	-1&\text{if }f_i(x)\ge f_i(x_0,\dots,x_j+1,\dots,x_{n-1})\text{ for all }x\in [s]^n\text{ with }x_j<s-1\\
	0&\text{otherwise}.
	\end{cases}
\]

Let $D$ be a signed digraph on $[n]$. We denote by $F(D,s)$ the set of networks on $D$, that is, the set of $f:[s]^n\to[s]^n$ such that $G(f)\subseteq D$. More explicitly, $f \in F(D,s)$ if and only if the following three constraints hold:
\begin{itemize}
	\item[-] $f_i$ depends on $x_j$ only if $(j,i)$ is an arc in $D$;
	\item[-] if $\lambda(j,i) = 1$, then $f_i$ is a non-decreasing function of $x_j$;
	\item[-] if $\lambda(j,i) = -1$, then $f_i$ is a non-increasing function of $x_j$.
\end{itemize}
We remark that $\lambda(j,i) = 0$ does not put any restriction as to how $f_i$ depends on $x_j$. For all $i\in [n]$, we define the order relation $\le_i$ on $\N^n$ by: 
\[
	x\le_i y\quad \iff\quad  x_{N^0(i)}=y_{N^0(i)}\text{ and }x_{N^+(i)}\le y_{N^+(i)}\text{ and }x_{N^-(i)}\ge y_{N^-(i)}. 
\]
This relation can then be used to characterise the set of networks on $D$ (the proof is a simple exercise).
\begin{lemma}\label{pro:monotonous}
For all $f:[s]^n\to [s]^n$ we have
\[
	f\in F(D,s)\qquad  \iff\qquad  \forall i\in [n],~\forall x,y\in [s]^n,\quad x\le_i y\Rightarrow f_i(x)\le f_i(y).
\]
\end{lemma}

%\begin{proof}
%Let $f \in F(D,s)$ and suppose $x \le_i y$ and $f_i(x) > f_i(y)$ for some $x, y \in [s]^n$ and $i \in [n]$. Then
%\end{proof}

The $s$-ary {\bf guessing number} $g(D,s)$ of a signed digraph $D$ is defined as the logarithm of the maximum number of fixed points in a network on $D$:
\[
	g(D,s)=\max_{f\in F(D,s)} \log_s |\Fix(f)|.
\]
The bounds on the guessing number known so far are
\begin{equation} \label{eq:old_bounds}
	c^+ \le g(D,s) \le k^+,
\end{equation}
where $k^+$ is the size of a minimum non-negative feedback vertex set in $D$ and $c^+$ is the maximum number of disjoint non-negative cycles in $D$ (we remark that $c^+ \le n/\gamma^+$); the upper bound is given in \cite{Richard09} and the lower bound is an easy exercise.

We finally remark that the case $s=2$ is referred to Boolean in the literature on Boolean networks, while  it is referred to as binary in coding theory and network coding. We shall use both terms interchangeably.

\subsection{Definition and general results} \label{sec:def_guessing}

The guessing graph was first proposed for unsigned digraphs in \cite{GR11} and then extended to closure operators (for the so-called closure solvability problem) in \cite{Gad13}. We now adapt the definition to the case of signed digraphs. Some of our results can be viewed as analogues of those in \cite{GR11}, while others are entirely novel.

\begin{definition}[Guessing graph of a signed digraph] \label{def:guessing_graph}
For any signed digraph $D$ on $[n]$ and any integer $s\ge 2$, the $s$-guessing graph of $D$, denoted as $\G(D,s)$, is the simple graph on $[s]^n$ with an edge $xy$ if there does not exist $f\in F(D,s)$ such that $x,y\in\Fix(f)$. 
\end{definition}

\begin{proposition}\label{pro:edges}
The set of edges of $\G(D,s)$ is 
\[
	E(\G(D,s))=\bigcup_{i\in[n]} E_i(D,s)
\]
with
\[
	E_i(D,s)=\{ xy :x,y\in[s]^n, \text{ either } x\le_i y\text{ and }~x_i>y_i\textrm{, or }y\le_i x\text{ and }~y_i>x_i\}.
\]
\end{proposition}

\begin{proof}
Suppose that $xy\in E(\G(D,s))$. Thus without loss of generality $x\le_i y$ and $x_i>y_i$ for some $i$. Let $f\in F(D,s)$. If $f(x)=x$, then by the monotony of $f_i$ we have $y_i<x_i=f_i(x)\le f_i(y)$. Thus $f$ cannot fix both $x$ and $y$ at the same time, i.e. $xy$ is an edge of $G(D,s)$. 

Conversely, suppose that $xy\not\in E(\G(D,s))$. Let $I$ be the set of $i\in [n]$ such that $x_i<y_i$. We define $f:[s]^n\to [s]^n$ as follows: for all $z\in [s]^n$,  
\[
	\forall i\in I,\quad 
	f_i(z)=
	\begin{cases}
	x_i&\text{ if }z\le_i x\\
	y_i&\text{ otherwise}
	\end{cases},
	\qquad 
	\forall i\in [n]\setminus I,\quad 
	f_i(z)=
	\begin{cases}
	y_i&\text{ if }z\le_i y\\
	x_i&\text{ otherwise}
	\end{cases}
\]
Let $i\in I$. Since $x_i<y_i$ and $xy\not\in E(\G(D,s))$, we have $y\not\le_i x$, thus $f_i(x)=x_i$ and $f_i(y)=y_i$. Furthermore, it is easy to see that $f_i$ is monotonous with respect to $\le_i$. Let $i\not\in I$, and suppose first that $y_i<x_i$. Since $xy\not\in E(\G(D,s))$, we have $x\not\le_i y$, thus $f_i(y)=y_i$ and $f_i(x)=x_i$. Furthermore, it is easy to see that $f_i$ is monotonous with respect to $\le_i$. Finally, if $x_i=y_i=c$ then $f_i=\cst=c$, and $f_i$ is trivially monotonous with respect to $\le_i$. Thus $f(x)=x$, $f(y)=y$, and by the monotony of the $f_i$, we have $f\in F(D,s)$. Thus there is no edge between $x$ and $y$ in $\G(D,s)$.     
\end{proof}

\begin{example} \label{ex:G}
The guessing graph of some classes of signed digraphs can be easily determined (the proof is an easy exercise).
\begin{itemize}
	\item If $D$ has a non-negative loop on each vertex, then $\G(D,s)$ is empty.
	
	\item For $C_n^+$, the directed cycle on $n$ vertices with all arcs signed positively,
	$$
		E(\G(C_n^+, s)) = \{xy: x < y \text{ or } x > y\}.
	$$
	This can be extended to any directed cycle with an even number of arcs signed negatively (and hence the cycle has positive sign).
	
	\item If $D$ is acyclic, then $\G(D,s)$ is complete; this is shown in \cite{GR11} for the unsigned case. We shall generalise this in Proposition \ref{prop:G_complete} below.
\end{itemize}
\end{example}

\begin{proposition} \label{prop:G_complete}
If $D$ has no non-negative cycle then $\G(D,s)$ is complete. 
\end{proposition}

\begin{proof}
Suppose that there is no edge between $x$ and $y$ in $\G(D,s)$, and let $I$ be the set of $i$ with $x_i\neq y_i$. For all $i\in I$ such that $x_i>y_i$ we have $x\not\le_i y$ thus there exists at least one vertex, say $i^*$, such that 
\[
i^*\in N^0(i)\text{ and }x_{i^*}\neq y_{i^*},
\quad\text{or}\quad 
i^*\in N^+(i)\text{ and }x_{i^*}> y_{i^*},
\quad\text{or}\quad 
i^*\in N^-(i)\text{ and }x_{i^*}< y_{i^*}.
\]
Similarly, for all $i\in I$ such that $y_i>x_i$ we have $y\not\le_i x$ thus a vertex $i^*$ such that
\[
i^*\in N^0(i)\text{ and }y_{i^*}\neq x_{i^*},
\quad\text{or}\quad 
i^*\in N^+(i)\text{ and }y_{i^*}> x_{i^*},
\quad\text{or}\quad 
i^*\in N^-(i)\text{ and }y_{i^*}< x_{i^*}.
\]
Note that we have the following property:
\begin{equation}\label{eq:lambda}
\forall i\in I,\qquad \lambda(i^*,i)\neq 0\quad\Rightarrow \quad \lambda(i^*,i)=\frac{\text{sign}{(x_i-y_i)}}{\text{sign}{(x_{i^*}-y_{i^*})}}.
\end{equation}

We are now in position to prove the proposition. Since $i^*\in I$ for all $i\in I$, there exists a cycle $C$ such that all the arcs of $C$ are of the form $(i^*,i)$. In other words, there exists a cycle $C=i_0,i_1,\dots,i_{k-1},i_0$ such that $i^*_{l+1}=i_l$ for all $l\in [k]$ (indices are computed modulo $k$). By hypothesis if some arc of this cycle has sign $0$, then $C$ is non-negative, a contradiction. Thus $C$ has no $0$ sign, and we deduce from (\ref{eq:lambda}) that the sign $\sigma$ of $C$ is 
\[
\sigma=\frac{\text{sign}{(x_{i_1}-y_{i_1})}}{\text{sign}{(x_{i_0}-y_{i_0})}}
\frac{\text{sign}{(x_{i_2}-y_{i_2})}}{\text{sign}{(x_{i_1}-y_{i_1})}}
\cdots
\frac{\text{sign}{(x_{i_{k-1}}-y_{i_{k-1}})}}{\text{sign}{(x_{i_{k-2}}-y_{i_{k-2}})}}
\frac{\text{sign}{(x_{i_0}-y_{i_0})}}{\text{sign}{(x_{i_{k-1}}-y_{i_{k-1}})}}\\
=1
\]
which is a contradiction.
\end{proof}

For any undirected unsigned simple graph $G$, we denote the independence number of $G$ as $\alpha(G)$.

\begin{theorem} \label{th:g=alpha}
For every non-empty independent set $Z$ of $\G(D,s)$ there exists $f\in F(D,s)$ such that $Z\subseteq \Fix(f)$, and hence $g(D,s)=\log_s\alpha(\G(D,s))$.
\end{theorem}

\begin{proof}
Foremost, by definition of the guessing graph, the set of fixed points of $f$ must form an independent set in the guessing graph, hence $g(D,s) \le \log_s \alpha(\G(D,s))$.

Conversely, for all $x\in [s]^n$, we set $Z(x,i)=\{z:z\in Z,~x\le_i z\}$ and we define $f:[s]^n\to [s]^n$ by 
\[
	\forall i\in [n],~\forall x\in [s]^n,\qquad 
	f_i(x)=\min(\{z_i:z\in Z(x,i)\}\cup\{\max(\{z_i:z\in Z\})\}).
\]  
If $x\le_i y$, then $Z(y,i)\subseteq Z(x,i)$ thus $f_i(x)\le f_i(y)$. Hence $f_i$ is monotonous with respect to $\le_i$, so $f\in F(D,s)$. Let $z\in Z$. Since $z\in Z(z,i)$ we have $f_i(z)\le z_i$. So if $f_i(z)\neq z_i$, there exists $y\in Z(z,i)$ such that $f_i(z)=y_i<z_i$. Thus we have $z\le_i y$ and $z_i>y_i$. But then according to the previous proposition, $\G(D,s)$ has an edge between $y$ and $z$, thus $Z$ is not an independent set, a contradiction. This means that $f_i(z)=z_i$ for all $i\in [n]$. Thus $f(z)=z$ for all $z\in Z$. 
\end{proof}

\begin{example} \label{ex:g}
The guessing number of some special digraphs can then be easily determined. Although the following are already known, they illustrate how to use the guessing graph approach to determine the guessing number.
\begin{itemize}
	\item If $D$ has a non-negative loop on each vertex, then $g(D,s) = n$; this is achieved by the identity function $f_i(x) = x_i$.

	\item If $D$ contains no non-negative cycle, then $g(D,s) = 0$. This was first proved in \cite{RC07}.
	
	\item For $C_n^+$, we obtain $g(C_n^+,s) = 1$; this is achieved by the function $f_i(x) = x_{i-1 \mod n}$.
\end{itemize}
\end{example}

%We finish this section with a couple of open questions. For an unsigned digraph $D$, it is well-known that 
%$$
	%\lim_{s \to \infty} g(D,s) = \sup_{s \ge 2} g(D,s) = H(D),
%$$
%the so-called entropy of the digraph \cite{Rii06, Rii07}.
%
%\begin{question}
%How does $g(D,s)$ behave for a signed graph $D$?
%\begin{enumerate}
	%\item Does the limit of $g(D,s)$ for $s$ tending to infinity always exist?
	%
	%\item Is there $D$ such that $g(D,s)>0$ for all $s \ge 2$, but with $g(D,s)$ tending to $0$ for $s$ tending to infinity?
	%
	%\item More generally, can $g(D,s)$ tend to a value $v$, where $g(D,s)>v$ for all $s$?
%\end{enumerate}
%\end{question}

We make the following remarks on the guessing graph of signed digraphs.
\begin{enumerate}
	\item Our definition is consistent with the guessing graph of unsigned digraphs introduced in \cite{GR11}. 
	
	\item If $D$ is a signed digraph and $D^0$ is the unsigned digraph with the same vertices and arcs as $D$, then $\G(D^0,s)$ is a spanning subgraph of $\G(D,s)$. Therefore, adding signs to a digraph can only reduce the number of fixed points.
	
	\item If $D$ is a signed digraph without any loops and $D'$ is obtained by adding a negative loop on some vertices of $D$, then $\G(D,s) = \G(D',s)$. Therefore, adding a negative loop on a vertex does not affect the sets of fixed points.
\end{enumerate}

\subsection{Bounds on the number of fixed points based on the guessing graph} \label{sec:bounds_guessing}

Based on the guessing graph, we shall derive bounds on the binary guessing number $g(D,2)$.

Let $D$ be a signed digraph on $[n]$ and $x\in [2]^n$. An arc $(j,i)$ of $D$ is \BF{$x$-frustrated} if $x_j\neq x_i$ and $\lambda(j,i)=1$ or $x_j= x_i$ and $\lambda(j,i)=-1$. We say that $D$ is $x$-frustrated if all its arcs are. Given $I\subseteq [n]$, we denote by $N(I,x)$ the set of vertices $i\in N(I)\setminus I$ such that all the arcs from $i$ to $I$ are $x$-frustrated. Note that if $D$ is unsigned, then $N(I,x)=\emptyset$ and $D[I]$ is $x$-frustrated if and only if $I$ is an independent set.

\begin{proposition}
The degree of a vertex $x$ in $\G(D,2)$ is 
\[
d(x)=
\sum_{
\begin{scriptsize}
\begin{array}{c}
I\subseteq [n]\\
D[I]\textrm{ is $x$-frustrated}
\end{array}
\end{scriptsize}}
(-1)^{|I|-1}2^{n-|N(I)\cup I|+|N(I,x)|}.
\]
\end{proposition}

\begin{proof}
By the inclusion-exclusion principle, we have 
\[
d(x)=|E(\G(D,2))\cap\{x\}|=\left|\bigcup_{i\in[n]}E_i(D,2)\cap\{x\}\right|=\sum_{I\subseteq [n]}(-1)^{|I|-1}|E_I(D,2)\cap\{x\}|
\]
where $E_I(D,2)=\bigcap_{i\in I} E_i(D,2)$ and hence we have only to determine $|E_I(D,2)\cap\{x\}|$ for all $I\subseteq [n]$. We have $xy\in E_I(D,2)$ if and only if $x\le_i y$ and $x_i>y_i$ or $y\le_i x$ and $y_i>x_i$ for all $i\in I$. Suppose that $xy\in E_I(D,2)$. Let $(j,i)$ be an arc of $D$ with $i\in I$. Suppose that $x_j=x_i$ and $\lambda(j,i)=1$. If $x_j\neq y_j$ then $x_j=x_i\neq y_i=y_j$ thus $x_i>y_i \Rightarrow x_j>y_j \Rightarrow x\not\le_i y$ and $y_i>x_i \Rightarrow y_j>x_j \Rightarrow y\not\le_i x$, a contradiction.  Thus
\[
x_j=y_j\text{ and }\lambda(j,i)=1\Rightarrow x_j=y_j
\]
and we prove similarly that 
\[
x_j\neq y_j\text{ and }\lambda(j,i)=-1\Rightarrow x_j=y_j.
\]
If $j\in I$ we have $x_j\neq y_j$ thus $(j,i)$ is $x$-frustrated. We deduce that $D[I]$ is $x$-frustrated. Suppose that $j\not\in I$. If $y_j\neq x_j$ then $(j,i)$ is $x$-frustrated thus $j\in N(I,x)$. Thus for all $i\in [n]$ we have $y_i\neq x_i$ if $i\in I$, we have $y_i=x_i$ if $i\in J=(N(I)\setminus I)\setminus N(I,x)$, and the component $y_i$ is free in the other cases. Thus $|E_I(D,2)\cap \{x\}|=2^{n-|I|-|J|}=2^{n-|N(I)\cup I|+|N(I,x)|}$. 
\end{proof}

We have 
\begin{multline*}
|E_i(D,s)\cap\{x\}|=s^{n-d_i-1}\Big((s-x_i)\prod_{j\in N^+(i)\setminus i}(x_j+1)\prod_{j\in N^-(i)\setminus i}(s-x_j+1)+\\x_i\prod_{j\in N^-(i)\setminus i}(x_j+1)\prod_{j\in N^+(i)\setminus i}(s-x_j+1)\Big)
\end{multline*}
thus there is no simple expression of $d(x)$ for $s>2$. 

Unlike the unsigned case, the guessing graph is not regular, let alone a Cayley graph. Therefore, some techniques used in \cite{GR11} for the guessing graph of unsigned digraphs cannot be applied here. We shall nonetheless derive two lower bounds on the guessing number. Both are based on the famous lower bound on the independence number of an undirected graph $G$ on $n$ vertices and with average degree $d$ (a corollary of Tur\'an's theorem \cite{Tol97}):
$$
	\alpha(G) \ge \frac{n}{d + 1}.
$$
Our first bound is a direct application of this bound for the whole guessing graph, while our second bound only considers a specific induced subgraph of the guessing graph.

%For any $I\subseteq [n]$ and $e\in [s]^{n-|I|}$ we denote by $\G(D,s)_I+e$ the subgraph of $\G(D,s)$ induced by the set of $x\in[s]^n$ such that $x_{[n]-I}=e$. 

%\begin{proposition}
%For every $I\subseteq [n]$ and $e\in [s]^{n-|I|}$, $\G(D,s)_I+e\simeq \G(D(i),s)$. 
%\end{proposition}

%For every $x\in[2]^n$, we denote by $0(x)$ the set of components $i\in [n]$ such that $x_i=0$ and we set $1(x)=[n]\setminus 0(x)$. 

\begin{theorem} \label{th:bound_g1}
We have 
\[
	g(D,2)\ge \delta^0 + \left(\log_2 \frac{4}{3} \right) \delta^\pm   - \log_2 n.
\] 
\end{theorem}

\begin{proof}
For all $x\in[2]^n$ and $i\in [n]$ we have 
\[
	|E_i(D,2)\cap \{x \}|=2^{n-d_i-1+|N(i,x)|},
\]
whence
\begin{align*}
	|E_i(D,2)|&=\frac{1}{2}\sum_{x}|E_i(D,2)\cap \{ x \}|\\[1mm]
	&=\frac{1}{2}\sum_{x} 2^{n-d_i-1+|N(i,x)|}\\[1mm]
	&=2^{n-d_i-2}\sum_{x} 2^{|N(i,x)|}.
\end{align*}
We need to evaluate the sum $\sum_{x} 2^{|N(i,x)|}$. Firstly, we remark that $|N(i,x)|$ does not depend on the value of $x$ outside of $N^+(i) \cup N^-(i)$. Moreover, for any $x$, let $z = (x_{N^+(i)} + x_i, x_{N^-(i)} + x_i + 1)$ with addition done componentwise, then $|N(i,x)|$ is the number of ones in $z$. Since there are ${d^+_i + d^-_i \choose k}$ choices for $z$ with exactly $k$ ones, there are $2^{n - d_i^+ - d_i^-} {d^+_i+d^-_i \choose k} $ states $x$ such that $|N(i,x)| = k$, whence
\begin{align*}
	|E_i(D,2)| 	&= 2^{n - d_i - 2} 2^{n - d_i^+ - d_i^-} \sum_{k=0}^{d^+_i + d^-_i}{d^+_i+d^-_i \choose k} 2^k\\[1mm]
	&=2^{2n - d_i^0 - 2d^+_i - 2d^-_i - 2} \, 3^{d^+_i+d^-_i}\\[1mm]
	&\le 2^{2n - \delta^0 -2} \left(\frac{3}{4}\right)^{\delta^\pm}.
\end{align*}
Thus
\begin{align*}
	d(\G(D,2)) &= \frac{2|E(\G(D,2))|}{2^n} =\frac{2|\cup_iE_i(D,2)|}{2^n} \le\frac{2\sum_i|E_i(D,2)|}{2^n}
	\le n2^{n- \delta^0 - 1} \left(\frac{3}{4}\right)^{\delta^\pm}\\
	\alpha(\G(D,2)) &\ge \frac{2^n}{d(G(D,2))+1}
	\ge \frac{2^n}{n2^{n-\delta^0} \left( \frac{3}{4} \right)^{\delta^\pm}}
	= 2^{\delta^0} \left( \frac{4}{3} \right)^{\delta^\pm} n^{-1}\\
	g(D,2) &= \log_2\alpha(\G(D,2))\ge \delta^0 +  \left(\log_2 \frac{4}{3} \right) \delta^\pm  - \log_2 n.
\end{align*}

\end{proof}

We remark that the bound above is smaller when there are more arcs in $D$ that are signed positively or negatively. In particular, if $D^0$ is unsigned then $\delta^0 =\delta$ and $\delta^\pm=0$ thus  
\[
	g(D^0,2) \ge \delta^0 - \log_2 n
\]
(see \cite{GR11}) and if $D^\pm$ is unate then $\delta^0 = 0$ and $\delta^\pm=\delta$ thus 
\[
	g(D^\pm,2) \ge \left( \log_2 \frac{4}{3} \right) \delta^\pm -\log_2 n.
\]
The bound in Theorem \ref{th:bound_g1} can be improved for the case of digraphs where most arcs are signed positively or negatively.

\begin{theorem} \label{th:bound_g2}
For any signed digraph $D$ with minimum in-degree $\delta\ge \frac{\ln(4n)}{2}$, 
$$
	g(D,2) \ge \frac{\delta}{2} - \sqrt{\frac{\ln(4n) \delta}{2}} - \log_2 n - 1.
$$
\end{theorem}

\begin{proof}
The main idea of the proof is to use a set $T$ of ``typical'' states $x \in [2]^n$ such that $T$ is large and the subgraph of $\G(D,2)$ induced by $T$ is sparser than the whole guessing graph. We denote $\epsilon := \frac{\ln(4n)}{2}$.

For any $i \in [n]$ and any $x \in [2]^n$, we have
$$
	|E_i(D,2) \cap \{x\}| = 2^{n-d_i-1 + |N(i,x)|}.
$$
For any $i \in [n]$, let
\begin{align*}
N^\pm(i)&:=(N^+(i)\cup N^-(i))\setminus\{i\}\\
d^\pm_i&:=|N^\pm(i)|\\
J_i &:= \{w : |w - \frac{d^\pm_i}{2}| <  \sqrt{\epsilon d^\pm_i}\},\\
	S_i(w) &:= \left\{x \in [2]^n : |N(i,x)| = w \right\},\\
	T_i &:= \bigcup_{w \in J_i} S_i(w) = \left\{ x \in [2]^n : \left| |N(i,x)| - \frac{d^\pm_i}{2} \right| < \sqrt{\epsilon d^\pm_i} \right\},\\
	T &:= \bigcap_{i=1}^n T_i.
\end{align*}

We first prove that $T$ is large. We have
\begin{align*}
	|S_i(w)| &= 2^{n-d^\pm_i} {d^\pm_i \choose w},\\
	|T_i| &= 2^n \sum_{w \in J_i} 2^{-d^\pm_i} {d^\pm_i \choose w}.
\end{align*}

Seeing each $x_i$ has a random variable following the Bernoulli distribution with parameter $1/2$, by Hoeffding's inequality \cite{Hoe63}, we have
\[
	\sum_{w\not\in J_i}2^{-d^\pm_i}{d^\pm_i \choose w}=\mathrm{Pr}\left(\left|\sum_{j\in N^\pm(i)}x_j-\frac{d^\pm_i}{2}\right|\geq \sqrt{	\epsilon d^\pm_i}\right)\leq2e^{-2\epsilon}
\]
thus
\begin{align*}
	|T_i| = 2^ n \sum_{w \in J_i} 2^{-d^\pm_i} {d^\pm_i \choose w} &\ge (1 - 2 e^{-2 \epsilon}) 2^n.
\end{align*}
By a simple recursion, we then prove that
$$
	|T| \ge (1-2ne^{-2\epsilon}) 2^n = 2^{n-1}.
$$

We now bound the degree of a vertex $x$ in $T$. We have
\begin{align*}
	d(x) &\le \sum_{i=1}^n |E_i(D,2) \cap \{x\}|\\
		&\le 2^{n-1} \sum_{i=1}^n 2^{-d_i+ d^\pm_i/2 + \sqrt{\epsilon d^\pm_i}}\\
	&\le 2^{n-1} \sum_{i=1}^n 2^{-d_i/2 + \sqrt{\epsilon d_i}}\\
	&\le n2^{n-1-\delta/2 + \sqrt{\epsilon \delta}}
\end{align*}
where we use the fact that $\delta\geq \epsilon$, hence the maximum term is with $d_i=\delta$. 

We then consider independent sets contained in $T$. We have
\begin{align*}
	\alpha(\G(D,2)) &\ge \frac{|T|}{n 2^{n-1 - \delta/2 + \sqrt{\epsilon \delta}} + 1}\\
	&\ge \frac{2^{n-1}}{n 2^{n - \delta/2 + \sqrt{\epsilon \delta}}}\\
	&= n 2^{\delta/2 - \sqrt{\epsilon \delta} - 1},\\
	g(D,2) &\ge \delta/2 - \sqrt{\epsilon \delta} - \log_2 n - 1.
\end{align*}
\end{proof}

\section{Error-correcting codes and signed networks} \label{sec:ECC_signed}

In this section, we investigate the properties of the set of fixed points $\Fix(f)$ for $f \in F(D,s)$. In particular, we see it as a code with special distance properties; these allow us to determine bounds on the maximum cardinality of $\Fix(f)$.

\subsection{Error-correcting codes} \label{sec:ECC}

An $s$-ary code $\mathcal{C}$ of length $n$ is simply a subset of $[s]^n$. The main parameter of $\mathcal{C}$ is its minimum distance: 
$$
	\mu_{\min}(\mathcal{C}) = \min \left\{ \mu(c,c') : c, c' \in \mathcal{C}, c \ne c' \right\},
$$
where $\mu$ is some distance function on $[s]^n$. We shall consider the following three distance functions. For any $x,y \in [s]^n$, let $L(x,y) := |\{i : x_i < y_i\}|$. The Hamming distance is defined as
$$
	\dH(x,y) := L(x,y) + L(y,x),
$$
i.e. it is the number of positions where $x$ and $y$ differ. The Max-distance is defined as
$$
	\dM(x,y) := \max\{L(x,y), L(y,x)\},
$$
while  the min-distance is defined as
$$
	\dm(x,y) := \min \{ L(x,y), L(y,x) \}.
$$
The reader who is interested in error-correcting codes with the Hamming distance is directed to the authoritative book by MacWilliams and Sloane \cite{MS77}. Binary codes based on the Max-distance were proposed for correcting asymmetric errors (it is called the asymmetric metric in the literature), such as those that occur in the Z-channel; for a review on these codes, see \cite{Klo95}. The min-distance is not a metric, which prevents the use of typical coding theory techniques. However, we will determine bounds on binary codes with the min-distance by relating them to codes with the Hamming distance.

The maximum cardinality of an $s$-ary code of length $n$ with minimum $\mu$ distance $d$ is denoted as $A_\mu(n,d,s)$. For the Hamming distance, this quantity has been widely studied, see \cite{Bro} for values and bounds for small parameter values. In particular, we have the Gilbert bound
\begin{equation} \label{eq:Gilbert}
	\AH(n,d,2) \ge \frac{2^n}{\sum_{k=0}^{d-1} {n \choose k}},
\end{equation}
the sphere-packing bound
\begin{equation} \label{eq:sphere-packing}
	\AH(n,d,2) \le \frac{2^n}{\sum_{k=0}^{\lfloor \frac{d-1}{2} \rfloor} {n \choose k}},
\end{equation}
and the Singleton bound $\AH(n,d,s) \le s^{n-d+1}$ which for $s=2$ is only attained in trivial cases and is usually much looser than the sphere-packing bound (however, the Singleton bound is tight for large alphabets). For codes with the Max-distance, only the binary case seems to have been studied. The Varshamov bound \cite{Var65} yields 
\begin{equation} \label{eq:Varshamov}
	\AM(n,d,2) \le \frac{2^{n+1}}{\sum_{j=0}^{d-1} {\lfloor n/2 \rfloor \choose j} + {\lceil n/2 \rceil \choose j} }.
\end{equation}
This is not the tightest bound known so far; see \cite{Klo95} for a review of upper bounds on $\AM(n,d,2)$.

The three distances are related in one way for general $s$ and two more ways if $s = 2$. The first relation (for all $s \ge 2$) between these distances simply follows their definitions:
$$
	\Am(n,d,s) \le \AH(n,2d,s) \le \AM(n,d,s).
$$
In the binary case ($s=2$), the second relation is given by the Borden bound \cite{Bor81}
$$
	\AM(n,d,2) \le d \AH(n,2d-1,2).
%	\AM(n,d,2) \le \AH(n+d-1,2d-1,2).
$$
In the binary case, the third way to relate the Max-distance and the min-distance to the Hamming distance is via the use of constant-weight codes. For any $x \in [s]^n$, let the {\em weight} of $x$ be $W(x) := \sum_{i=0}^{n-1} x_i$ (in $\mathbb{N}$). For any weight $0 \le w \le n(s-1)$, we denote the set of states with weight $w$ as
$$
	B(n,w,s) := \{x \in [s]^n : W(x) = w\}.
$$
A (binary) constant-weight code $\mathcal{C}$ of length $n$ and weight $w$ is simply a subset of $[2]^n$ where all the codewords in $\mathcal{C}$ have weight $w$. We denote the maximum cardinality of a constant-weight code of length $n$, weight $w$ and minimum distance $d$ as $\AH(n,d,w,2)$. Due to their many applications and great theoretical interest, constant-weight codes have been thoroughly studied, see for instance \cite{BSSS90, AVZ00} and \cite{Broa} for a table for small parameter values. In particular, we shall use the Bassalygo-Elias bound \cite{Bas65}
\begin{equation} \label{eq:Bassalygo}
	\AH(n,d,w,2) \ge \frac{{n \choose w}}{2^n} \AH(n,d,2).
\end{equation}
Since $2\dm(x,y) = \dH(x,y)$ for all $x,y \in [2]^n$ with equal weight, we obtain
$$
	\Am (n,d,2) \ge \AH \left(n, 2d, \left\lfloor \frac{n}{2} \right\rfloor, 2 \right).
$$

Our bounds on the guessing number will include binomial coefficients, which can be approximated as follows \cite[Chapter 10]{MS77}. We denote the binary entropy function as 
$$
	H(p) := -p \log_2 p - (1-p) \log_2(1-p)
$$
for $p \in [0,1]$; then 
\begin{gather} 
	\label{eq:binomial}
	\frac{2^{nH(\lambda)}}{\sqrt{8n \lambda (1- \lambda)}}  \le {n \choose \lambda n} \le \frac{2^{nH(\lambda)}}{\sqrt{2 \pi n \lambda (1- \lambda)}},\\
	\label{eq:sum_binomial}
	\frac{2^{nH(\mu)}}{\sqrt{8n \mu (1- \mu)}}  \le \sum_{k=0}^{\mu n} {n \choose k} \le 2^{nH(\mu)},
\end{gather}
provided $\lambda n$ is an integer between $0$ and $n$ and $\mu n$ an integer between $0$ and $n/2$; in particular
\begin{equation} \label{eq:central_binomial}
	{n \choose \lfloor n/2 \rfloor} \ge \frac{2^{n-1}}{\sqrt{2n}}.
\end{equation}

Let us summarise below our remarks on the quantities we have introduced so far in the binary case.

\begin{lemma} \label{lem:AH_Anabla}
We have
$$
	\frac{1}{2 \sqrt{2n}} \AH (n, 2d, 2) \le \AH \left(n, 2d, \left\lfloor \frac{n}{2} \right\rfloor, 2 \right) 
	\le \Am (n,d,2) 	\le \AH(n,2d,2) 
	%\le \frac{2^n}{{n \choose \left\lfloor \frac{n}{2} \right\rfloor}} \AH \left(n, 2d, \left\lfloor \frac{n}{2} \right\rfloor , 2 \right) 
	%\le 2\sqrt{2n} \AH \left(n, 2d, \left\lfloor \frac{n}{2} \right\rfloor , 2 \right)
%$$
%and
%$$
%	\AH(n,2d,2) 
\le \AM(n,d,2) \le d \AH(n,2d-1,2).
$$
\end{lemma}

\subsection{Bounds on the guessing number for all digraphs} \label{sec:ECC_bounds}

\begin{theorem} \label{th:g<A}
For any signed digraph $D$, and any $f \in F(D,s)$, $\Fix(f)$ is a code of length $n$ with minimum Hamming distance at least $\gamma^+$. Thus,
$$
	g(D,s) \le \log_s \AH(n,\gamma^+,s).
$$
Moreover, for any negative digraph $D^-$, we have
$$
	g(D^-,s) \le \log_s \Am \left( n, \frac{\gamma^+}{2}, s \right);
$$
and for any positive digraph $D^+$, we have
$$
	g(D^+,s) \le \log_s \AM (n, \gamma^+ , s).
$$
\end{theorem}

\begin{proof}
The proof is based on the same argument as that of Proposition \ref{prop:G_complete}. If $x$ and $y$ are distinct and not adjacent in the guessing graph $\G(D,s)$, then let $I$ be the set of positions $i$ where $x_i \ne y_i$. For any $i \in I$, we have $x_{N(i)} \ne y_{N(i)}$ hence there exists $j \in I \cap N(i)$. Applying this fact repeatedly, we obtain that $i$ belongs to a cycle in the digraph induced by $I$. This cycle must be non-negative, since otherwise we would have $x_i < y_i$ and $x_i > y_i$. Thus $\dH(x,y) \ge \gamma^+$.

Moreover, if $D^+$ is a positive digraph, we see that if $x_i < y_i$, then $x_k < y_k$ on all the vertices in $I$; thus $L(x,y) \ge \gamma^+$. If instead $x_i > y_i$, then $L(y,x) \ge \gamma^+$; in any case $\dM(x,y) \ge \gamma^+$.

Finally, if $D^-$ is a negative digraph, then in order for the cycle to be non-negative, it must have even length and we see that the sign of $x_k - y_k$ alternates on the cycle; thus $\dm(x,y) \ge \gamma^+/2$.
\end{proof}

%\begin{proof}
%For any subset $H$ of vertices and any $e \in [s]^{n - |H|}$, we denote the subgraph of $\G(D,s)$ induced by all the configurations satisfying $x_{V \setminus H} = e$ as $\G(D,s)_{\dH} + e$.
%
%\begin{claim*}
%For any $H$ and $e$, $\G(D,s)_{\dH} + e$ is isomorphic to $\G(D[H],s)$.
%\end{claim*}
%
%\begin{proof}[Proof of Claim]
%Two configurations $x,y$ are adjacent in $\G(D,s)_{\dH} + e$ if and only if there exists $h \in H$ such that $x \le_{\dH} y$ in $D$ and $x_{\dH} > y_{\dH}$ (or vice versa). Since $x_{V \setminus H} = y_{V \setminus H}$, this is equivalent to $x \le_{\dH} y$ in $D[H]$ and $x_{\dH} > y_{\dH}$ (or vice versa).
%\end{proof}
%
%Let $x,y \in [2]^n$ and suppose that $I := \{i : x_i \ne y_i\}$ induces a graph without any non-negative cycles. Since $x_{V \setminus I} = y_{V \setminus I}$ and $\mathrm{G}(D,2)_I + x_{V \setminus I}$ is complete by Proposition \ref{prop:G_complete}, we obtain that $x$ and $y$ are adjacent in $\mathrm{G}(D,2)$. 
%
%Therefore, if $a$ and $b$ are not adjacent in $\mathrm{G}(D,2)$, then they must disagree on a set which contains at least one non-negative cycle and hence $\dH(a,b) \ge \gamma^+$. In other words, any set of fixed points (i.e. any independent set of the guessing graph) must form a code of minimum distance at least $\gamma^+$.
%\end{proof}

%We remark that the proof of Theorem \ref{th:g<A} does not actually use the fact that the alphabet is binary. We can then extend the result so that $g(D,s)$ is no more than the maximum cardinality of a code in $[s]^n$ with minimum distance $\gamma^+$.

\begin{corollary}[Sphere-packing bound for the guessing number]
For any signed directed graph $D$, let $t := \frac{1}{n}\left\lfloor \frac{\gamma^+-1}{2} \right\rfloor$, then
\begin{align*}
	g(D,2) &\le n - \log_2 \sum_{s=0}^{nt} {n \choose s}\\
	&\le n  - nH(t) + \frac{1}{2} \log_2 n + \frac{1}{2}  \log_2 (8t(1-t)).
\end{align*}
\end{corollary}

\begin{theorem} \label{th:g>A}
Let $D$ be a signed directed graph. We denote
$$
	\phi := \max_{1 \le i \le n} \min \left\{ \frac{n-d_i^0 + 1}{2}, n - d_i^0 - d_i^- + 1, n - d_i^0 - d_i^+ + 1 \right\}.
$$
Any code with minimum min-distance at least $\phi$ is a subset of fixed points of some $f \in F(D,s)$. Thus
$$
	g(D,s) \ge \log_s \Am (n,\phi,s),
$$
and in particular
$$
	g(D,2) \ge \log_2 \Am (n,\phi,2) \ge  \log_2 \AH(n,2 \phi, \lfloor n/2 \rfloor,2).
$$
\end{theorem}

\begin{proof}
Suppose that $xy \in E_i(D,s)$ for some $i$, say $x \le_i y$ and $x_i > y_i$. We must have $x_{N^0(i)} = y_{N^0(i)}$ and hence 
$$
	2 \dm(x,y) \le \dH(x,y) \le n - d_i^0.
$$
We also have $x_j \le y_j$ for all $j \in N^0(i) \cup N^+(i)$, hence 
$$
	\dm(x,y) \le L(y,x) \le n- d_i^0 - d_i^+;
$$
similarly 
$$
	\dm(x,y) \le n - d_i^0 - d_i^-.
$$
The conjunction of these three conditions implies $\dm(x,y) < \phi$. Thus any code with minimum min-distance at least $\phi$ forms an independent set of $\G(D,s)$.
\end{proof}

By combining the Bassalygo-Elias bound and the Gilbert bound, we then obtain another lower bound on $g(D,2)$, which is usually tighter for graphs with high minimum in-degree.

\begin{corollary}[Gilbert bound for the guessing number] \label{cor:g>A}
We have
\begin{align*}
	g(D,2) &\ge \log_2 {n \choose \left\lfloor n/2 \right\rfloor} - \log_2 \sum_{k=0}^{2 \phi - 1} {n \choose k}\\
	&\ge n - n H \left( \min \left\{ \frac{2 \phi - 1}{n}, \frac{1}{2} \right\} \right) - \frac{1}{2} \log_2 n - \frac{3}{2}.
\end{align*}
\end{corollary}

We finish this section with an open question about the guessing number. For unsigned digraphs, the limit of the guessing number always exists (we shall say more about it later); however, this remains open in the case of signed digraphs.

\begin{question} \label{q:g}
Does $\lim_{s \to \infty} g(D,s)$ exist for any signed digraph $D$?
\end{question}

\section{Functions defined over signed cliques} \label{sec:positive}

\subsection{Refined bounds for positive or negative functions} \label{sec:bounds_positive}

We are now interested in fixed points of functions whose signed interaction graphs are fully and equally signed, i.e. either all arcs are signed positively (positive function) or negatively (negative function). In view of the remarks above, we only consider digraphs without loops; therefore we are interested in the guessing numbers of $K_n^+$, the positive clique on $n$ vertices, and of $K_n^-$, the negative clique on $n$ vertices. First of all, their respective guessing graphs can be easily determined.

\begin{lemma} \label{lem:G(Kn)}
\begin{enumerate}	
	\item For $K_n^-$, we have
	$$
		E(\G(K_n^-,s)) = \{xy : x \le y \text{ or } y \le x\} = \{xy: \dm(x,y) = 0\}.
	$$
	
	\item For $K_n^+$, we have
	$$
		E(\G(K_n^+,s)) = \{xy : L(x,y) = 1 \text{ or } L(y,x) = 1\}.
	$$
\end{enumerate}
\end{lemma}

\begin{proof}
For $K_n^-$, $xy$ is an edge in the guessing graph if and only if there exists $i$ such that $x_i < y_i$ and $x_j \le y_j$ for all $j \ne i$ (or vice versa), which is equivalent to $x \le y$ (or $y \le x$).

For $K_n^+$, $xy$ is an edge in the guessing graph if and only if there exists $i$ such that $x_i < y_i$ and $x_j \ge y_j$ for all $j \ne i$ (or vice versa), which is equivalent to $L(x,y) = 1$ (or $L(y,x) = 1$).
\end{proof}

\begin{proposition} \label{prop:Kn-}
For $K_n^-$, we have
$$
	g(K_n^-,s) = \log_s \left | B\left(n, \left\lfloor \frac{n(s-1)}{2} \right\rfloor, s \right) \right|, 
$$
and in particular
\begin{align*}
	g(K_n^-,2) &= \log_2 {n \choose \left\lfloor \frac{n}{2} \right\rfloor}
	\ge n - \frac{1}{2}\log_2 n - \frac{3}{2}\\
	\lim_{s \to \infty} g(K_n^-,s) &= \sup_{s \ge 2} g(K_n^-,s) = n-1 =k^+.
\end{align*}
\end{proposition}

\begin{proof}
A set of fixed points of $f\in F(K_n^-,s)$ is a set of incomparable states, i.e. an antichain according to the partial order $\le$. This partial order is isomorphic to the lattice of divisors of $N = (p_1 p_2 \dots p_n)^{s-1}$, where $p_i$ is the $i$-th prime number. Thus, the largest antichain is $B(n, n(s-1)/2, s)$, a result called Sperner's theorem for $s=2$ and then extended to any $s$ in \cite{DVK51}. The bound on the binomial coefficient then follows from \eqref{eq:central_binomial}. For the limit, we observe that $W(x)$ takes a value between $0$ and $n(s-1)$, hence
\begin{align*}
	\max_{0 \le w \le n(s-1)} |B(n,w,s)| &\ge \frac{s^n}{n(s-1) + 1}\\
	\log_s \max_w |B(n,w,s)| &\ge n - 1 - o(1),
\end{align*}
while $g(K_n^-,s) \le k^+= n-1$.
\end{proof}

%\begin{proposition} \label{prop:Kn+Delta}
%Let $\mathcal{C}$ be the set of fixed points of a positive function without loops in its interaction graph. Then $\mathcal{C}$ has minimum asymmetric distance at least $2$ and hence
%$$
	%g(K_n^+,2) \le \log_2 A_\Delta(n,2) \le n + 1 - \log_2 (n+2).
%$$
%\end{proposition}
%
%\begin{proof}
%If $x$ and $y$ are distinct and not adjacent in the guessing graph $\G(K_n^+,2)$, then $\Delta(x,y) \ge 2$.
%\end{proof}

%\begin{proposition} \label{prop:Cn+Delta}
%Let $D^+$ be a completely positively signed digraph. Then if $x$ and $y$ are both fixed by the same function in $F(D,s)$, then $L(x,y) = 0 $ or $L(x,y) \ge \gamma^+$ (or vice versa). Thus
%$$
	%g(D^+,2) \le \log_2 A_\Delta(n,\gamma^+).
%$$
%\end{proposition}
%
%\begin{proof}
%If $x$ and $y$ are distinct and not adjacent in the guessing graph $\G(D^+,2)$, then let $I$ be the set of positions $i$ where $x_i < y_i$. For any $i \in I$, we have $x_{N(i)} \le y_{N(i)}$ and these two must be distinct thus there exists $j \in I \cap N(i)$. This clearly forces $I$ to induce a graph with a cycle.
%\end{proof}

Although the guessing graph of $K_n^+$ is easy to determine, its independence number (and hence the guessing number of $K_n^+$) is still unknown. The bounds on the guessing number reviewed in \eqref{eq:old_bounds} yield 
\begin{equation}\label{eq:classicbounds}
	\left\lfloor \frac{n}{2} \right\rfloor \le g(K_n^+,s) \le n-1.
\end{equation}
We shall significantly improve on those bounds. Firstly, we consider the binary case.

\begin{proposition} \label{prop:Kn+2}
The binary guessing number of $K_n^+$ satisfies
$$
		n - \frac{3}{2}\log_2 n - \frac{3}{2} \le g(K_n^+,2) \le n - \log_2 (n+2) + 1.
$$
\end{proposition}

\begin{proof}
For the lower bound, Theorem \ref{th:g>A} yields
%\begin{align*}
	%g(K_n^+,2) &\ge \log_2 \AH(n,4,\left \lfloor \frac{n}{2} \right\rfloor, 2)\\
	%&\ge \log_2 {n \choose \left \lfloor \frac{n}{2} \right\rfloor} - \log_2 n,\\
	%&\ge n - \frac{3}{2}\log_2 n - \frac{3}{2},
%\end{align*}
$$
	g(K_n^+,2) \ge \log_2 \AH(n,4,\left \lfloor \frac{n}{2} \right\rfloor, 2)
	\ge \log_2 {n \choose \left \lfloor \frac{n}{2} \right\rfloor} - \log_2 n
	\ge n - \frac{3}{2}\log_2 n - \frac{3}{2},
$$

where the second inequality comes from the construction in \cite[Theorem 1]{GS80}, which we shall adapt in the proof of Proposition \ref{prop:Kn+s}. For the upper bound, Theorem \ref{th:g<A} together with the Varshamov bound in \eqref{eq:Varshamov} yield
$$
	g(K_n^+,2) \le \log_2 \AM(n,2,2) \le \log_2 \frac{2^{n+1}}{n+2}.
$$
\end{proof}

%In comparison, the usual lower bound for the number of fixed points of positive functions, i.e. with all arcs in $G(f)$ signed positively, without loops in their graphs is achieved by packing positive cycles. We can pack $\lfloor n/2 \rfloor$ cycles, thus yielding $2^{\lfloor n/2 \rfloor}$ fixed points, which is way below the bound. On the other hand, the obvious upper bound on the number of fixed points is $2^{k^+} = 2^{n-1}$ which is refined below.
%
%\begin{proposition} \label{prop:upper_bound_Kn+}
%Let $\mathcal{C}$ be the set of fixed points of a positive function without loops in its interaction graph. Then
%$$
	%|\mathcal{C}| \le \frac{2^{n+1}-1}{n+1}.
%$$
%\end{proposition}
%
%\begin{proof}
%We may assume that $\mathcal{C}$ is the entire set of fixed points of a unate function $f$. Let $\mathcal{C}_w$ denote the set of fixed points with weight $w$. Then this forms a constant-weight code of Hamming distance at least $4$. Indeed, if $c,c' \in \mathcal{C}_w$ are at Hamming distance $2$, then $L(c,c') = L(c',c) = 1$, and hence they are adjacent in the guessing graph.
%Thus 
%\begin{align*}
	%|\mathcal{C}| &= \sum_{w=0}^n |\mathcal{C}_w|\\
	%&\le \sum_{w=0}^n \AH(n,4,w)\\
	%&\le \sum_{w=0}^n \frac{1}{w} {n \choose w-1}\\
	%&= \sum_{w=0}^n \frac{1}{n+1} {n+1 \choose w}\\
	%&= \frac{2^{n+1}-1}{n+1}.
%\end{align*}
%The upper bound on $\AH(n,4,w)$ comes from \cite[Problem (6), Chapter 17]{MS77}.
%\end{proof}

We now investigate general alphabets.

\begin{proposition} \label{prop:Kn+s}
For $K_n^+$, we have
\begin{align*}
	g(K_2^+,s) &= 1 \quad \forall s \ge 2,\\
	g(K_3^+,s) &= \log_s \left( \left\lfloor \frac{3(s-1)}{2} \right\rfloor + 1 \right) \quad \forall s \ge 2\\
	\lim_{s \to \infty} g(K_3^+,s) &= \inf_{s \ge 2} g(K_3^+,s) = 1,\\
	\lim_{s \to \infty} g(K_4^+,s) &= 2,
\end{align*}
and for all $n \ge 4$,
$$
	n-3 \le \limsup_{s \to \infty} g(K_n^+,s) \le n-2.
$$
\end{proposition}

\begin{proof}
Foremost, we have $g(K_n^+,s) \le n-1$ for any $n$ and $s$. Since $K_2^+$ is a positive cycle, we have a guessing number of 1.

For $K_3^+$, suppose $x$ and $y$ are not adjacent in the guessing graph and let $W(x) \le W(y)$. Then we claim that $x \le y$ and $W(y) \ge W(x) + 2$. Indeed, denote the coordinates as $i$, $j$, and $k$: since $y$ has higher weight, we have $x_i < y_i$ and since they are not adjacent, $x_j < y_j$; by non-adjacency we must then have $x_k \le y_k$. Thus, any independent set in the guessing graph is a chain of length at most $\lfloor 3(s-1)/2 \rfloor + 1$.

Conversely, construct the following infinite chain recursively. Let $c^0 = (0,0,0)$ and for any $k \ge 0$,
$$
	c^{k+1} = c^k + \begin{cases}
	(1,1,0) &\text{if } k \equiv 0 \mod 3\\
	(0,1,1) &\text{if } k \equiv 1 \mod 3\\
	(1,0,1) &\text{if } k \equiv 2 \mod 3.
	\end{cases}
$$
The sequence starts
$$
	c^0 = (0,0,0), \, c^1 = (1,1,0), \, c^2 = (1,2,1), \, c^3 = (2,2,2), \, c^4 = (3,3,2) \dots
$$
It is easy to check that the first $\lfloor 3(s-1)/2 \rfloor + 1$ terms in the sequence belong to $[s]^3$ and that for any $k < l$, $L(c^k, c^l) \ge 2$ and $L(c^l, c^k) = 0$. Therefore, these terms form an independent set in the guessing graph $\G(K_3^+,s)$.

For $K_4^+$, the lower bound in \eqref{eq:classicbounds} yields $g(K_4^+,s) \ge 2$ for all $s \ge 2$; the limit follows from the upper bound on $\limsup_{s \to \infty} g(K_n^+,s)$, which we now prove. Let $\mathcal{C}$ be the largest set of fixed points of a network on $K_n^+$ ($n \ge 4$), and let $\mathcal{C}_w$ be the set of codewords in $\mathcal{C}$ with weight $w$. If $x,y \in \mathcal{C}_w$ are distinct, we have $L(x,y) > 0$  since they have equal weight and hence $L(x,y) \ge 2$, since they belong to $\mathcal{C}$. Similarly, $L(y,x) \ge 2$ which yields $\dH(x,y) \ge 4$. By the Singleton bound, $|\mathcal{C}_w| \le s^{n-3}$ and hence
$$
	g(K_n^+,s) = \log_s |\mathcal{C}| \le \log_s \left\{ (n(s-1) + 1) s^{n-3} \right\} = n-2 + o(1).
$$

We now prove the lower bound on the limit of the guessing number of $K_n^+$. For any $m = (m_0,m_1,m_2) \in \mathbb{N}^3$, define the code $\mathcal{C}_m \subseteq [s]^n$ as
$$
	\mathcal{C}_m = \{x \in [s]^n : W_0(x) = m_0, W_1(x) = m_1, W_2(x) = m_2\},
$$
where
\begin{align*}
	W_0(x) &:= \sum_{i=1}^n x_i = W(x),\\
	W_1(x) &:= \sum_{i=1}^n i x_i,\\
	W_2(x) &:= \sum_{i=1}^n i^2 x_i.
\end{align*}
We claim that for any $x,y \in \mathcal{C}_m$, $\dm(x,y) \ge 2$ (and in particular, they are not adjacent in the guessing graph). First, since $W(x) = W(y)$, we have $\dm(x,y) \ge 1$. Suppose then that $L(x,y) = 1$, i.e. $y_j > x_j$ for a unique position $j$ and $x_b > y_b$ for $b \in B \subseteq [n] \backslash \{j\}$. To clarify notation, let us denote $Y_j := y_j - x_j > 0$ and $X_b := x_b - y_b > 0$. We have
\begin{align*}
	\sum_{b \in B} X_b &= Y_j\\
	\sum_{b \in B} b X_b &= j Y_j\\
	\sum_{b \in B} b^2 X_b &= j^2 Y_j.
\end{align*} 
In particular, we have $\sum_b (b-1) X_b = (j-1) Y_j$ and $\sum_b (b-1)^2 X_b = (j-1)^2 Y_j$, hence we can shift the sequences until $j=1$ (and we have negative and positive indices). We obtain
$$
	\sum_b X_b = \sum_b bX_b = \sum_b b^2 X_b = Y_1.
$$
Then $\sum_b (b^2 - b) X_b = 0$, which implies that only $X_0$ must be nonzero ($X_1$ does not exist since $j=1$) and hence $\sum_b b X_b = 0 < Y_1$ which is the desired contradiction.

Now, for all $x \in [s]^n$ and $0 \le i \le 2$, $0 \le W_i(x) \le n^{i+1}(s-1)$, thus
\begin{align*}
	\max_m |\mathcal{C}_m| &\ge s^n \left\{ \big(n(s-1) + 1\big) \big( n^2(s-1) + 1 \big) \big( n^3 (s-1) + 1 \big) \right\}^{-1}\\
	\log_s g(K_n^+,s) &\ge n - 3 - o(1).
\end{align*}

\end{proof}

We provide two remarks on those results. Firstly, recall that for an unsigned digraph $D^0$, the limit of the guessing number is its supremum:
$$
	\lim_{s \to \infty} g(D^0,s) = \sup_{s \ge 2} g(D^0,s) = H(D^0),
$$
the so-called entropy of the digraph \cite{Rii06, Rii07}. However, $K_3^+$ is an example where this is completely reversed, for
$$
	g(K_3^+,s) > \lim_{s \to \infty} g(K_3^+,s) \qquad \forall \, s \ge 3.
$$
Therefore, the guessing number of signed digraphs can exhibit some behaviour which cannot be seen in unsigned digraphs. Interestingly, the negative clique $K_n^-$ does behave like an unsigned digraph since the limit of the guessing number is indeed its supremum.

Secondly, some results in the literature tend to suggest that non-negative cycles tend to produce many fixed points. This is reflected in our upper bound on the guessing number in Theorem \ref{th:g<A}, which directly depends on the non-negative girth. However, $D_1 = K_{n+1}^+$ and $D_2 = K_n^-$ are two signed digraphs such that $D_1$ has more non-negative cycles and more disjoint non-negative cycles than $D_2$, while $D_2$ has a higher guessing number. Therefore, the guessing number is not always an increasing function of the number of (disjoint) non-negative cycles. This goes against the common view mentioned above, and somehow echoes a result in \cite{ARS14} on the number of fixed points of conjunctive networks. These are Boolean networks where every local update function $f_i(x)$ is a conjunction of literals: a positive or negative sign on the arc $(j,i)$ indicates whether the literal is $x_j$ or $\neg x_j$. It is shown in \cite{ARS14} that the maximum number of fixed points of a conjunctive network without loops in its interaction graph is obtained by using a disjoint union of triangles, where all arcs are signed negatively. Therefore, maximising the number of fixed points in the conjunctive case goes against maximising the number of (disjoint) positive cycles.

We finish this section with a natural open question, given the gap in Proposition \ref{prop:Kn+2}.

\begin{question} \label{q:Kn+}
What is $\lim_{s \to \infty} g(K_n^+,s)$ for $n \ge 5$, if it exists?
\end{question}

\subsection{Convergence for positive or negative functions} \label{sec:convergence}

Our combinatorial approach based on guessing graphs and coding theory completely forgets about the actual networks with a given set of fixed points. Interestingly, sometimes a given set of fixed points $S \subseteq [s]^n$ admits a network $f \in F(D,s)$ which is easy to describe and hence to analyse. One main property we would like to study is whether the network  actually {\bf converges} to $S$, i.e. if for any $x$, there exists a positive integer $k$  such that $f^k(x) \in S$.

We first prove that convergence to a set of fixed points of maximal size can never occur for networks in $F(K_n^-,s)$.

\begin{proposition} \label{prop:Kn-CV}
For any $0 \le w \le n(s-1)$, $B(n,w,s)$ is the set of fixed points of the function in $F(K_n^-,s)$ defined by
$$
	f_i(x) = \mathrm{saturation} \left( w - \sum_{j \ne i} x_j \right),
$$
where saturation is a function from $\mathbb{Z}$ to $[s]$ defined as
$$
	\mathrm{saturation}(a) := \begin{cases}
	0 &\text{if } a < 0\\
	a &\text{if } 0 \le a \le s-1\\
	s-1 &\text{if } a > s-1.
	\end{cases}
$$

However, no function in $F(K_n^-,s)$ can converge to $B(n,w,s)$ for any $s-1 \le w \le (n-1)(s-1)$; in particular for $n \ge 2$ no such function converges to the largest set of fixed points $B(n, \lfloor n(s-1)/2 \rfloor, s)$.
\end{proposition}

\begin{proof}
For the function defined above, we have $x = f(x)$ if and only if $x_i = w - \sum_{j \ne i} x_j$ for all $i$, which is equivalent to $W(x) = w$. Now, suppose $f \in F(K_n^-,s)$ converges to $B(n,w,s)$, where $s-1 \le w \le (n-1)(s-1)$. Then consider $f(0,\dots,0)$: for any $x$ and any $i$, $f_i(x) \le f_i(0,\dots,0)$; however since there always exists a state $x \in B(n,w,s)$ such that $x_i = s-1$, we must have $f_i(0,\dots,0) = s-1$ for all $i$. Similarly, we have $f(s-1,\dots,s-1) = (0,\dots,0)$, thus these two states form an asymptotic cycle.
\end{proof}

For $K_n^+$, we are unable to describe a maximum set of fixed points. However, we have seen that in the  binary case, an optimal constant-weight code will be nearly optimal. We further remark that if $\mathcal{C}$ is a constant-weight code of minimum Hamming distance at least $4$ and weight $w \ge 2$, then $\mathcal{C}' := \mathcal{C} \cup \{(0,\dots,0), (1,\dots,1)\}$ is the set of fixed points of a positive function $f \in F(K_n^+,2)$, thus
$$
	g(K_n^+,2) \ge \log_2 \left( \AH(n,4,\left \lfloor \frac{n}{2} \right\rfloor, 2) + 2 \right).
$$ 
As such, we investigate convergence towards such a code. We prove that for most constant-weight codes of minimum distance $4$, there is a positive function which converges to $\mathcal{C}'$ in only three time steps. 

\begin{proposition} \label{prop:Kn+CV}
For any constant-weight code $\mathcal{C}$ of length $n$, weight $w$, and minimum Hamming distance $4$, such that $3 \le w \le n-3$ and $2w \ne n$, the following function $f \in F(K_n^+,2)$ converges towards $\mathcal{C}'$ in three steps.

For any vertex $i$ and any $a \in [s]$, we denote $x_{-i} := x_{V \setminus \{i\}}$ and we use the shorthand notation $(a,x_{-i}) := (x_0, \dots, x_{i-1}, a, x_{i+1}, \dots, x_{n-1})$; then
$$
	f_i(x_{-i}) = \begin{cases}
		1 &\text{if } (1,x_{-i}) \in \mathcal{C} \text{ or } W(x_{-i}) \ge w+1 \text{ or } (W(x_{-i}) = w \text{ and } (0,x_{i}) \notin \mathcal{C}),\\
		0 &\text{if } (0,x_{-i}) \in \mathcal{C} \text{ or } W(x_{-i}) \le w-2 \text{ or } (W(x_{-i}) = w-1 \text{ and } (1,x_{i}) \notin \mathcal{C}).
	\end{cases}
$$
\end{proposition}

\begin{proof}
It is clear that $G(f) = K_n^+$. Let us prove that $f^3(x)\in\mathcal{C}'$ by cases on $x$.
\begin{enumerate}
	\item \label{it:case1} $W(x) \le w-2$. Then $W(x_{-i}) \le w-2$ for all $i$ and hence $f(x) = (0,\dots,0)$.
		
	\item \label{it:case2} $W(x) = w-1$ and $(1,x_{-i}) \in \mathcal{C}$ for some $i$. Firstly, remark that $i$ is unique: if $\dH((1,x_{-j}), (1,x_{-i})) = 2$, hence $(1,x_{-j}) \notin \mathcal{C}$. We then have $f_i(x) = 1$ and $f_k(x) = 0$ for any other coordinate $k$ (if $x_k = 0$, $W(x_{-k}) = w-1$ and $(1,x_{-k}) \notin \mathcal{C}$; if $x_k = 1$, $W(x_{-k}) = w-2$). Since $W(f(x)) = 1 \le w-2$, Case \ref{it:case1} yields $f^2(x) = (0,\dots,0)$.
	
	\item \label{it:case3} $W(x) = w-1$ and $(1,x_{-i}) \notin \mathcal{C}$ for all $i$. Then $f(x) = (0,\dots,0)$ (if $x_k = 0$, $W(x_{-k}) = w-1$ and $(1,x_{-k}) \notin \mathcal{C}$; if $x_k = 1$, $W(x_{-k}) = w-2$).
	
	\item \label{it:case4} $x \in \mathcal{C}$. Then $f(x) = x$.
	
	\item \label{it:case6} $W(x) = w+1$ and $(0,x_{-i}) \notin \mathcal{C}$ for all $i$. Then similarly to Case \ref{it:case3}, we obtain $f(x) = (1,\dots,1)$.

	\item \label{it:case7} $W(x) = w+1$ and $(0,x_{-i}) \in \mathcal{C}$ for some $i$. Then similarly to Case \ref{it:case2}, we obtain $f^2(x) = (1,\dots,1)$.
	
	\item \label{it:case8} $W(x) \ge w + 2$. Similarly to Case \ref{it:case1}, we obtain $f(x) = (1,\dots,1)$.
	
	\item \label{it:case5} $W(x) = w$ and $x \notin \mathcal{C}$. Then we obtain $f(x) = x + (1,\dots,1)$. If $w < n/2$, we have $W(f(x)) \ge w+1$ and hence we are in Case \ref{it:case6} to \ref{it:case8}; otherwise we have $W(f(x)) \le w -1$ and we are in Case \ref{it:case1} to \ref{it:case3}. Altogether, we obtain $f^3(x) \in \{(0,\dots,0), (1,\dots,1)\}$.
\end{enumerate}
\end{proof}

\section{Comparison between bounds} \label{sec:comparison_bounds}

The bounds we have determined so far are difficult to compare for they depend on different parameters of the digraph and on the alphabet size; moreover, some are only valid for certain classes of signed digraphs. For the sake of clarity, we shall only consider $s=2$ because our results would then be valid for Boolean networks and also because we have more relations between the different distances, and only positive digraphs $D^+$ since this is an important special case, and for which we obtain the tightest bounds.

In order to compare different bounds, we shall use their asymptotic behaviour. This is a technique commonly used in coding theory \cite{MS77}, where we investigate a sequence of binary codes $\mathcal{C}_n$ of length $n$, minimum distance $d_n$ (for the distance $\mu \in \{\dH, \dm, \dM\}$), and  such that $|\mathcal{C}_n| = A_\mu(n, d_n, 2)$. We consider the asymptotic notation
\begin{align*}
	\bar{d} &:= \lim_{n \to \infty} \frac{d_n}{n},\\
	\bar{A_\mu}(\bar{d}) &:= \limsup_{n \to \infty} \frac{\log_2 A_\mu(n, n\bar{d}, 2)}{n},
\end{align*}
and investigate how the asymptotic rate $\bar{A_\mu}$ behaves as a function of $\bar{d}$. For instance, the Gilbert bound in \eqref{eq:Gilbert}, together with the estimates on sums binomial coefficients in \eqref{eq:sum_binomial}, yield
$$
	\bar{\AH}(\bar{d}) \ge 1 - H(\bar{d}),
$$
which is the tightest asymptotic bound known so far. On the other hand, the Singleton bound yields $\bar{\AH}(\bar{d}) \le \bar{d}$, which is well below the asymptotic version of the sphere-packing bound in \eqref{eq:sphere-packing}:
$$
	\bar{\AH}(\bar{d}) \le 1 - H(\bar{d}/2).
$$
In fact, the sphere-packing bound is not the tightest asymptotic upper bound known so far. Instead, the celebrated McEliece-Rodemich-Rumsey-Welch (MRRW) bound yields \cite{MS77}
$$
	\bar{\AH}(\bar{d}) \le \mathrm{MRRW}(\bar{d}),
$$
where
\begin{align*}
	\mathrm{MRRW}(\bar{d}) &= \min_{0 < u \le 1 - 2 \bar{d}} \left\{ 1 + h(u^2) - h(u^2 + 2 \bar{d}u + 2 \bar{d}) \right\},\\
	h(x) &= H \left( \frac{1}{2} - \frac{1}{2} \sqrt{1-x} \right).
\end{align*}
In particular, $\mathrm{MRRW}(\bar{d}) = 0$ if $\bar{d} \ge 1/2$.

%H \left( \frac{1}{2} - \sqrt{\bar{d}(1-\bar{d})} \right)

Our results on different distances in Lemma \ref{lem:AH_Anabla} then show that
$$
	\bar{\Am}(\bar{d}) = \bar{\AM}(\bar{d}) = \bar{\AH}(2\bar{d}).
$$

In order to study the asymptotic behaviour of the guessing number, we need to introduce some asymptotic notation for digraph parameters. Let $\{D^+_n\}$ be a sequence of positive digraphs on $n$ vertices with minimum vertex feedback set of size $k^+_n$, positive girth $\gamma^+_n$, and minimum in-degree $\delta^+_n$, where
\begin{align*}
	\lim_{n \to \infty} \frac{g(D^+_n,2)}{n} &= \bar{g},\\
	\lim_{n \to \infty} \frac{k^+_n}{n} &= \bar{k},\\
	\lim_{n \to \infty} \frac{\gamma^+_n}{n} &= \bar{\gamma},\\
	\lim_{n \to \infty} \frac{\delta^+_n}{n} &= \bar{\delta}.
\end{align*}
Using this notation, our bounds can readily be translated to asymptotic form.

\begin{proposition} \label{prop:asymptotic}
Asymptotically, we have two competing lower bounds
\begin{align*}
	\bar{g} &\ge \frac{1}{2} \bar{\delta},\\
	\bar{g} &\ge \bar{\AH}(2(1- \bar{\delta})) \ge 1 - H(2(1 - \bar{\delta})) \quad \text{for } \bar{\delta} \ge \frac{3}{4},
\end{align*}
and two competing upper bounds
\begin{align*}
	\bar{g} &\le \bar{\AH}(2\bar{\gamma}) \le \mathrm{MRRW}(2 \bar{\gamma}),\\
	\bar{g} &\le \bar{k}.
\end{align*}
\end{proposition}

\begin{proof}
The first lower bound is an immediate translation of Theorem \ref{th:bound_g2}. Also, for a positive digraph $D^+$, we have $\phi = n - \delta^+ + 1$, hence Theorem \ref{th:g>A} yields 
$$
	\bar{g} \ge \bar{\Am}(1- \bar{\delta}) = \bar{\AH}(2(1- \bar{\delta}));
$$
the asymptotic Gilbert bound gives the second lower bound $1 - H(2(1- \bar{\delta}))$. The remaining two lower bounds are asymptotically looser than $\bar{\delta}/2$: the $c^+$ lower bound in \eqref{eq:old_bounds} is upper bounded by $n / \gamma^+$, which asymptotically yields zero, while Theorem \ref{th:bound_g1} only yields $\bar{\delta} \log_2 (4/3)$.

For upper bounds, Theorem \ref{th:g<A} yields 
$$
	\bar{g} \le \bar{\AM}(\bar{\gamma}) = \bar{\AH}(2\bar{\gamma});
$$
the MRRW bound then yields the first upper bound. Finally, the feedback vertex set bound immediately yields the second upper bound $\bar{g} \le \bar{k}$.
\end{proof}

The lower bounds are easy to compare for they depend on the same parameter $\bar{\delta}$. They are displayed in Figure \ref{fig:lower}. In other to properly compare the upper bounds, we show in Figure \ref{fig:upper} the values of $\bar{\gamma}$ for which the MRRW bound is tighter than the feedback bound for a given value of $\bar{k}$. We first remark that since $\gamma^+ \le n - k^+ + 1$ for any signed digraph $D$, we have $\bar{\gamma} \le 1 - \bar{k}$; hence $\bar{\gamma}$ always lies below the top curve. For $\bar{\gamma}$ anywhere between the two curves, the bound from Theorem \ref{th:g<A} is tighter than the feedback vertex set bound; in particular, if $\bar{\gamma} \ge 1/4$, then Theorem \ref{th:g<A} yields $\bar{g} = 0$. We can then conclude that the $k^+$ bound is usually weaker than  the coding-theoretic bound in Theorem \ref{th:g<A}. This can be intuitively explained by the fact that $k^+$ corresponds to the Singleton bound for codes with minimum Hamming distance $n - k^+ +1$; however, as we mentioned earlier the Singleton bound tends to be poor for the binary case. Thus, unless $n - k^+ + 1$ is significantly higher than $\gamma^+$, the $k^+$ bound will be loose.

\begin{figure}
	\centering
	\includegraphics[scale=0.5]{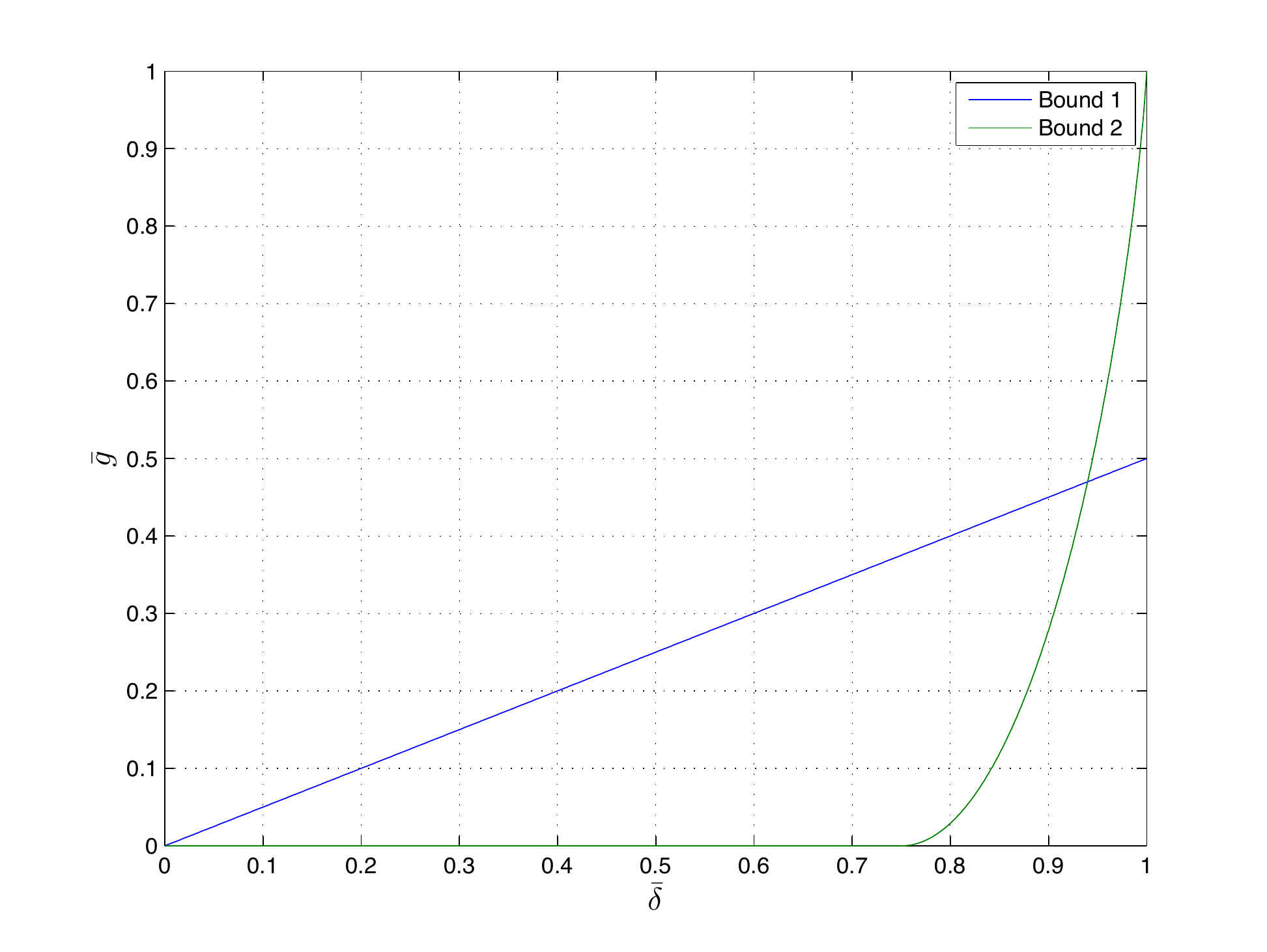}
	\caption{Asymptotic lower bounds on the guessing number as a function of the minimum in-degree.} \label{fig:lower}
\end{figure}

\begin{figure}
	\centering
	\includegraphics[scale=0.5]{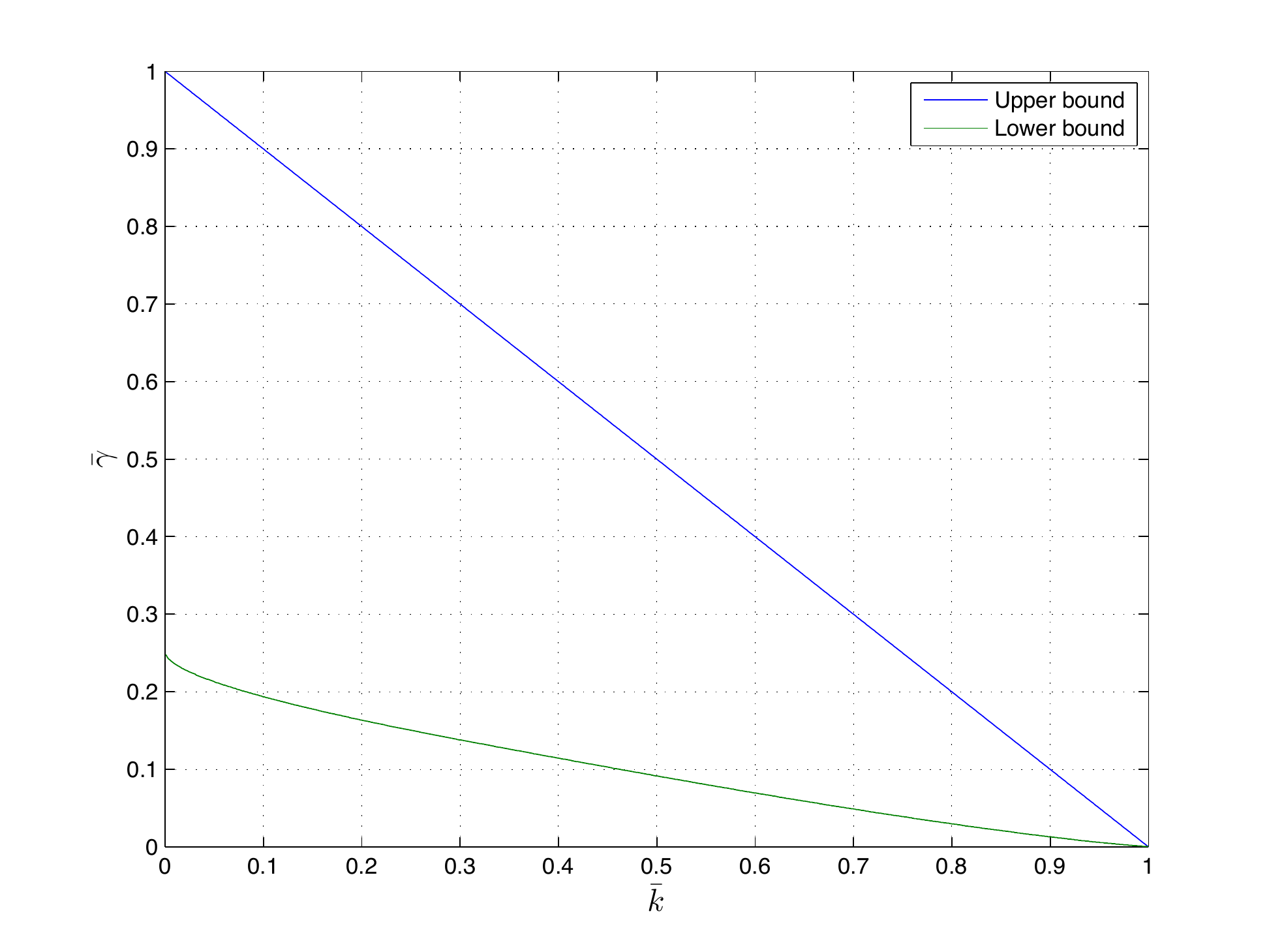}
	\caption{$\bar{\gamma}$ compared to $\bar{k}$: anywhere between the two curves shows an improvement over the $\bar{k}$ asymptotic upper bound.} \label{fig:upper}
\end{figure}

%\bibliographystyle{amsplain}
%\bibliography{g}

\begin{thebibliography}{10}

\bibitem{AVZ00}
E. Agrell, A. Vardy, and K. Zeger, \emph{Upper bounds for
  constant-weight codes}, IEEE Transactions on Information Theory \textbf{46}
  (2000), no.~7, 2373--2395.

\bibitem{ACLY00}
R.~Ahlswede, N.~Cai, S.-Y.~R. Li, and R.~W. Yeung, \emph{Network information
  flow}, IEEE Transactions on Information Theory \textbf{46} (2000), no.~4,
  1204--1216.

\bibitem{Ara08}
J. Aracena, \emph{Maximum number of fixed points in regulatory {B}oolean
  networks}, Bulletin of mathematical biology \textbf{70} (2008), 1398--1409.

\bibitem{Aracena2004}
J. Aracena, J. Demongeot, and Eric Goles, \emph{Fixed points and
  maximal independent sets in {AND-OR} networks}, Discrete Applied Mathematics
  \textbf{138} (2004), no.~3, 277--288.

\bibitem{ADG04}
\bysame, \emph{Positive and negative circuits in discrete neural networks},
  IEEE Transactions on Neural Networks \textbf{15} (2004), no.~1, 77--83.

\bibitem{ARS14}
J.~Aracena, A.~Richard, and L.~Salinas, \emph{Maximum number of fixed points in
  {AND-OR-NOT} networks}, Journal of Computer and System Sciences \textbf{80}
  (2014), no.~7, 1175--1190.

\bibitem{Bas65}
L.~A. Bassalygo, \emph{New upper bounds for error correcting codes}, Problems
  of Information Transmission \textbf{1} (1965), no.~4, 32--35.

\bibitem{Bor81}
J.~M. Borden, \emph{Bounds and constructions for error correcting/detecting
  codes on the {Z}-channel}, Proceedings of IEEE International Symposium on
  Information Theory, 1981, pp.~94--95.

\bibitem{Broa}
A.~E. Brouwer, \emph{Bounds for binary constant weight codes},
  \url{http://www.win.tue.nl/~aeb/codes/Andw.html}.

\bibitem{Bro}
\bysame, \emph{Table of general binary codes},
  \url{http://www.win.tue.nl/~aeb/codes/binary-1.html}.

\bibitem{BSSS90}
A.~E. Brouwer, J.~B. Shearer, N.~J.~A. Sloane, and W.~D. Smith, \emph{A
  new table of constant weight codes}, IEEE Transactions on Information Theory
  \textbf{36} (1990), no.~6, 1334--1380.

\bibitem{DVK51}
N.~G. de~Bruijn, Ca. van Ebbenhorst~Tengbergen, and D.~Kruyswijk, \emph{On the
  set of divisors of a number}, Nieuw Arch. Wiskunde \textbf{23} (1951),
  191--193.

\bibitem{Gad13}
M. Gadouleau, \emph{Closure solvability for network coding and secret
  sharing}, IEEE Transactions on Information Theory (2013), no.~12, 7858--7869.

\bibitem{GR11}
M. Gadouleau and S. Riis, \emph{Graph-theoretical constructions
  for graph entropy and network coding based communications}, IEEE Transactions
  on Information Theory \textbf{57} (2011), no.~10, 6703--6717.

\bibitem{Gol85}
E. Goles, \emph{Dynamics of positive automata networks}, Theoretical Computer
  Science \textbf{41} (1985), 19--32.

\bibitem{Goles90}
E. Goles and S. Mart\'{\i}nez, \emph{Neural and automata networks:
  Dynamical behavior and applications}, Kluwer Academic Publishers, Norwell,
  MA, USA, 1990.

\bibitem{GT83}
E. Goles and M.~Tchuente, \emph{Iterative behaviour of generalized majority
  functions}, Mathematical Social Sciences \textbf{4} (1983), 197--204.

\bibitem{GS80}
R.~L. Graham and N.~J.~A. Sloane, \emph{Lower bounds for constant weight
  codes}, IEEE Transactions on Information Theory \textbf{26} (1980), no.~1,
  37--43.

\bibitem{Hoe63}
W. Hoeffding, \emph{Probability inequalities for sums of bounded random
  variables}, Journal of the American Statistical Association \textbf{58}
  (1963), 13--30.

\bibitem{Hop82}
J.~Hopfield, \emph{Neural networks and physical systems with emergent
  collective computational abilities}, Proc. Nat. Acad. Sc. U.S.A. \textbf{79}
  (1982), 2554--2558.

\bibitem{Jong02}
H.~De Jong, \emph{Modeling and simulation of genetic regulatory systems: A
  literature review}, Journal of Computational Biology \textbf{9} (2002),
  67--103.

\bibitem{KS08}
G. Karlebach and R. Shamir, \emph{Modelling and analysis of gene regulatory
  networks}, Nature \textbf{9} (2008), 770--780.

\bibitem{Kau69}
S.~A. Kauffman, \emph{Metabolic stability and epigenesis in randomly connected
  nets}, Journal of Theoretical Biology \textbf{22} (1969), 437--467.

\bibitem{Klo95}
T. Kl\o{}ve, \emph{Error correcting codes for the asymmetric channel},
  1995, \url{http://www.ii.uib.no/~torleiv/rap95.pdf}.

\bibitem{MP43}
W.~S. {Mac Culloch} and W.~S. Pitts, \emph{A logical calculus of the ideas
  immanent in nervous activity}, Bull. Math. Bio. Phys. \textbf{5} (1943),
  113--115.

\bibitem{MS77}
F.~J. MacWilliams and N.~J.~A. Sloane, \emph{The theory of error-correcting
  codes}, North-Holland, Amsterdam, 1977.

\bibitem{PS83}
S. Poljak and M. Sura, \emph{On periodical behaviour in societies
  with symmetric influences}, Combinatorica \textbf{3} (1983), 119--121.

\bibitem{RRT08}
E. Remy, P. Ruet, and D. Thieffry, \emph{Graphic requirements for
  multistability and attractive cycles in a {B}oolean dynamical framework},
  Advances in Applied Mathematics \textbf{41} (2008), no.~3, 335--350.

\bibitem{Richard09}
A.~Richard, \emph{Positive circuits and maximal number of fixed points in
  discrete dynamical systems}, Discrete Applied Mathematics \textbf{157}
  (2009), no.~15, 3281--3288.

\bibitem{Ric13}
\bysame, \emph{Fixed point theorems for {B}oolean networks expressed in
  terms of forbidden subnetworks}, March 2013, \url{http://arxiv.org/abs/1302.6346}.

\bibitem{RC07}
A.~Richard and J.-P. Comet, \emph{Necessary conditions for multistationarity in
  discrete dynamical systems}, Discrete Applied Mathematics \textbf{155}
  (2007), no.~18, 2403--2413.

\bibitem{Rii06}
S. Riis, \emph{Utilising public information in network coding}, General
  Theory of Information Transfer and Combinatorics, Lecture Notes in Computer
  Science, vol. 4123/2006, Springer, 2006, pp.~866--897.

\bibitem{Rii07a}
\bysame, \emph{Graph entropy, network coding and guessing games}, November
  2007, \url{http://arxiv.org/abs/0711.4175}.

\bibitem{Rii07}
\bysame, \emph{Information flows, graphs and their guessing numbers}, The
  Electronic Journal of Combinatorics \textbf{14} (2007), 1--17.

\bibitem{R86}
F.~Robert, \emph{Discrete iterations: a metric study}, Series in Computational
  Mathematics, vol.~6, Springer, 1986.

\bibitem{SD05}
M.-H. Shih and J.-L. Dong, \emph{A combinatorial analogue of the
  {J}acobian problem in automata networks}, Advances in Applied Mathematics
  \textbf{34} (2005), 30--46.

\bibitem{Tho73}
R.~Thomas, \emph{{B}oolean formalization of genetic control circuits}, Journal
  of Theoretical Biology \textbf{42} (1973), no.~3, 563--585.

\bibitem{Tho80}
\bysame, \emph{On the relation between the logical structure of systems and
  their ability to generate multiple steady states or sustained oscillations},
  Spriner Series in Synergies \textbf{9} (1980), 180--193.

\bibitem{TK01a}
R.~Thomas and M.~Kaufman, \emph{Multistationarity, the basis of cell
  differentiation and memory. {II}. {L}ogical analysis of regulatory networks
  in terms of feedback circuits}, Chaos: An Interdisciplinary Journal of
  Nonlinear Science \textbf{11} (2001), no.~1, 180--195.

\bibitem{TD90}
R. Thomas and R. D'Ari, \emph{Biological feedback}, CRC Press, 1990.

\bibitem{Tol97}
L. Tolhuizen, \emph{The generalized {Gilbert–-Varshamov} bound is implied by
  {T}ur\'an's theorem}, IEEE Transactions on Information Theory \textbf{43}
  (1997), 1605--1606.

\bibitem{Var65}
R.~R. Varshamov, \emph{Some features of linear codes that correct asymmetric
  errors}, Soviet Physics-Doklady \textbf{9} (1965), 538--540.

\bibitem{YLCZ06}
R.~W. Yeung, S.-Y.~R. Li, N. Cai, and Z. Zhang, \emph{Network
  coding theory}, vol.~2, Foundation and Trends in Communications and
  Information Theory, no. 4-5, now Publishers, Hanover, MA, 2006.

\end{thebibliography}

\providecommand{\bysame}{\leavevmode\hbox to3em{\hrulefill}\thinspace}
\providecommand{\MR}{\relax\ifhmode\unskip\space\fi MR }
\providecommand{\MRhref}[2]{%
  \href{http://www.ams.org/mathscinet-getitem?mr=#1}{#2}
}
\providecommand{\href}[2]{#2}

\end{document}